\documentclass[12pt]{article}
\usepackage{cite}
\usepackage{amsmath,amsthm, amscd, amssymb, amsfonts}
\usepackage{appendix}
\usepackage{graphicx}
\newtheorem{theorem}{\textbf{Theorem}}
\begin{document}
\title{Bell inequalities, Counterfactual Definiteness and Falsifiability}
\author{Justo Pastor Lambare}
\date{}
\maketitle
\begin{abstract}
We formally prove the existence of an enduring incongruence pervading a widespread interpretation of the Bell inequality and explain how to rationally avoid it with a natural assumption justified by explicit reference to a mathematical property of Bell's probabilistic model. Although the amendment does not alter the relevance of the theorem regarding local realism, it brings back Bell theorem from the realm of philosophical discussions about counterfactual conditionals to the concrete experimental arena.
\end{abstract}


%
\section{Introduction}\label{sec:intro}
The purpose of this letter is not so much to contribute to the debate between quantum localists and non-localists, but to advocate for the logical consistency of the Bell theorem, correcting a widespread view that spoils such consistency and which has produced unnecessary confusion and unjustified criticisms \cite{pKup15,pCza20,pLam20b}.

The problem dates back to around 1971 when Henry Pierce Stapp\cite{pSta71} introduced counterfactual reasoning to ague for the nonlocal character of quantum theory.
H. P. Stapp produced a mathematical contradiction showing the incompatibility of quantum mechanical predictions with locality assumptions.
However, he did not obtain a directly falsifiable inequality as Bell originally proposed\cite{pBel64} and was later generalized by Clauser, Horne, Shimony, Holt (CHSH)\cite{pCHSH69}.

Unfortunately, after Stapp's paper, many physicists and philosophers began to assume that counterfactual reasoning is either necessary or correct to derive the CHSH form of the Bell inequality(BI)\cite{pEbe77,pPer78,pSky82,pFor86,pGil14,pZuk14,pKup15,pWol15,pBou17,pJCh19,pCza20}.

The idea became so widespread that even a new term was invented for this kind of counterfactual speculation used to prove the BI, \emph{counterfactual definiteness} (CFD).
We recall that CFD can be defined  as: \emph{``This assumption allows one to assume the definiteness of the results of measurements, which were actually not performed on a given individual system. They are treated as unknown, but in principle defined values.
This is in a striking disagreement with quantum mechanics, and the complementarity principle.''}\cite{pZuk14} CFD is believed to be implied by realism or be equivalent to it. It is supposed \emph{``to be the factor separating classical from quantum theories''}\cite{pHan19e}.

Contrary to common beliefs, the issues associated with CFD to prove the BI are not at all related to differences between classical physics and quantum mechanics peculiarities.
The difficulty is related to the naive application of methods valid for the formal sciences in a context that is characteristic of the factual sciences.

Despite philosophical or interpretational opinions, we show that the physical untenability of CFD -- if we are going to consider the BI as a meaningful and falsifiable result -- can be objectively and mathematically proved.
\section{Derivation of the Bell Inequality}\label{sec:deriv}
We succinctly review the derivation of the deterministic version of the CHSH form of the Bell inequality.
We shall consider Bell's and CHSH's standard inequality derivation, as well as the counterfactual version.
\subsection{The Standard Derivation}\label{ssec:BSD}
The main assumptions are locality, measurement independence(MI), and realism.
Realism is taken in the Einstein, Podolsky, and Rosen(EPR)\cite{pEPR35} sense as the existence of the ``elements of physical reality'' that we can predict with certainty.
In this case, they are the values of the spin given by the local realistic functions $A(a,\lambda)$ and $B(b,\lambda)$, assuming values $\pm1$.

John Bell and EPR considered that realism\footnote{Bell did not mention realism. Realism is a contentious concept\cite{pNor07}. In our case, we take realism as defined above and is tantamount to determinism. This view is appropriate for our discussion.} is a consequence of locality\cite{pMau14a,pNor07,pLau18a,pLau18b}, while others consider it an independent
assumption\cite{pZuk14,pWis14,pWer14,pGil14}, however, this polemic is not important for our argument.
Here we assume realism as an independent hypothesis.

MI is the assumption that the distribution function of the hidden variables $p(\lambda)$ is independent of the device setting variables.
The correlation term of Alice and Bob joint measurements of the singlet state is given by
{\small
\begin{equation}\label{eq:cv}
    E(a_i,b_k) = \int p(\lambda)\,A(a_i,\lambda)B(b_k,\lambda)\,d\lambda\,;\quad i,k\in\{1,2\}
\end{equation}
}
By adequately adding the correlation terms
\begin{eqnarray}
S    &=& E(a_1,b_1)-E(a_1,b_2)+E(a_2,b_1)+E(a_2,b_2)\label{eq:s1}\\
     &=& \int p(\lambda)\,C(\lambda)\,d\lambda\label{eq:s2}\\
|S|  &\leq& \int p(\lambda)\,|C(\lambda)|\,d\lambda\label{eq:s3}\\
     &\leq& \int p(\lambda)\, 2\,d\lambda\label{eq:s4}\\
     &\leq& 2\int p(\lambda)\,d\lambda\label{eq:s5}\\
     &\leq& 2\label{eq:s6}
\end{eqnarray}
The term $C(\lambda)$ in (\ref{eq:s2}) is given by
{\small
\begin{eqnarray}
  C(\lambda) &=& A(a_1,\lambda)B(b_1,\lambda)-A(a_1,\lambda)B(b_2,\lambda)+A(a_2,\lambda)B(b_1,\lambda)\nonumber\\
             &=&+A(a_2,\lambda)B(b_2,\lambda)\label{eq:cs} 
\end{eqnarray}
}
The last equation is crucial for the derivation and a frequent source of bewilderment because it is necessary to have the same value of $\lambda$ in the four addends of (\ref{eq:cs}) to properly factorize the equation and find the bound of 2 for $|S|$. Although we only discuss deterministic hidden variable models, basically the same problem is present in stochastic models.
\subsection{The Counterfactual Derivation}\label{ssec:CFDD}
Counterfactual reasoning allows a straightforward and elegant Bell inequality derivation.
Perhaps, that is the reason for its widespread acceptance, notwithstanding its inconsistency, which we shall make evident in sec. \ref{sec:TPICFD}.

The necessary hypothesis are CFD, locality and freedom.
We shall discuss freedom in the next section.

The argument is the following.
Let $A_i$ and $B_k$; $i,k\in\{1,2\}$ be Alice's and Bob's results when they jointly measure their particles under settings $a_i$ and $b_k$.
Assuming Alice and Bob actually measure $A_1$ and $B_1$, then by locality if Bob would have used setting $2$ instead of $1$, Alice's result could not have changed and we obtain the counterfactual result $A_1B_2$.
Applying a similar reasoning we find the other two counterfactual results and we can write
\begin{eqnarray}
s &=& A_1B_1 - A_1B_2 + A_2B_1 + A_2B_2\label{eq:cfe1}
\end{eqnarray}
We can apply the same argument if, instead of actually measuring $A_1$ an $B_1$, Alice and Bob would have actually measured any of the other three possibilities.
The important point is that (\ref{eq:cfe1}) is according to Stapp's counterfactual description: \emph{``Of these eight numbers only two can be compared directly to experiment. The other six correspond to the three alternative experiments that could have been performed but were not''}.\footnote{It is fair to say that, in his 1971 article\cite{pSta71}, H. P. Stapp did not inconsistently employ counterfactual reasoning. He did not apply it to obtain (\ref{eq:cfe1}) and derive the BI.}

Considering that $A_i$ and $B_k$ take only two values $\pm1$, (\ref{eq:cfe1}) becomes
\begin{eqnarray}
s &=& A_1(B_1 - B_2) + A_2(B_1 + B_2)\label{eq:cfe2}\\
s &=&\pm2\label{eq:cfe3}
\end{eqnarray}
taking mean values in (\ref{eq:cfe1}) according to (\ref{eq:cfe2}) and (\ref{eq:cfe3}) and putting $S=\langle s\rangle$
\begin{equation}\label{eq:cfe4}
\left.
\begin{array}{rclc}
-2 \leq S &= & \langle A_1B_1 - A_1B_2 + A_2B_1 + A_2B_2\rangle                                      & \leq 2\\
-2\leq S  &= & \langle A_1B_1\rangle - \langle A_1B_2\rangle + \langle A_2B_1\rangle + A_2B_2\rangle & \leq 2\\
|S|       &\leq& 2                                                                                   &
\end{array}
\right\}
\end{equation}
The CFD derivation has the advantage that the hidden variables become unnecessary as well as the MI hypothesis.
\subsection{Freedom and the Bell Inequality}\label{ssec:FBI}
The freedom of the experimenters to choose their settings is a fundamental hypothesis for obtaining the Bell inequality.
Both the standard derivation and the counterfactual version need that hypothesis.
\subsubsection{Measurement Independence and Freedom}
The fact that the distribution function $p$ does not contain the experimenters' settings parameters -- i.e., MI -- was an ad hoc hypothesis in Bell's 1964 derivation\cite{pBel64}.

MI, and its relation to freedom, was recognized by Bell as an independent hypothesis in 1975\cite{pBel76} and further discussed specially after a criticism of Shimony, Horne, and Clauser\cite{pSHC76}.

If $p(a,b)$ is Alice's and Bob's joint probability of choosing their settings, their freedom is defined by
\begin{equation}\label{eq:free}
p(a,b|\lambda)=p(a,b)
\end{equation}
According to Bayes's theorem of probability theory
\begin{equation}\label{eq:bayes1}
p(a,b,\lambda)=p(a,b|\lambda)p(\lambda)=p(\lambda|a,b) p(a,b)
\end{equation}
From (\ref{eq:free}) and (\ref{eq:bayes1}) we obtain MI
\begin{equation}\label{eq:mi}
p(\lambda|a,b)=p(\lambda)
\end{equation}
According to our definition, freedom(F) and MI -- (\ref{eq:free}) and (\ref{eq:mi}) -- are equivalent but we only need the implication
\begin{equation}\label{leq:fmi}
F\rightarrow MI
\end{equation}
Let R and L stand for realism and locality, the Bell theorem can be expressed as
\begin{equation}\label{leq:SBT}
R \wedge L \wedge F\rightarrow BI
\end{equation}
\subsubsection{CFD and Freedom}
The BI counterfactual version does not assume hidden variables, so freedom, as (\ref{eq:free}) defines it, is not necessary.
However, some liberty is still necessary to derive the counterfactual equation (\ref{eq:cfe1})\cite{pEbe77,pGil14}.
We must be able to choose counterfactually between setting possibilities.

The freedom necessary to derive (\ref{eq:cfe1}) is less stringent than the one required by MI.
To allow a counterfactual election of $a$ and $b$, instead of requiring (\ref{eq:free}),  we only need those possibilities to exist for possible values of $\lambda$
\begin{equation}\label{eq:cf}
p(a,b|\lambda)\neq 0
\end{equation}
Let us call this freedom \emph{counterfactual freedom}(CF), then the Bell theorem according to CFD  is
\begin{equation}\label{leq:CBT}
CFD \wedge L \wedge CF\rightarrow BI
\end{equation}
(\ref{eq:free}) already implies CF since it would not make sense to speak of freedom when one of the possible settings is not available to  an experimenter
\begin{equation}\label{leq:fandcf}
F\rightarrow CF
\end{equation}
The fact that experimenters may still have freedom despite the violation of (\ref{eq:mi}) is discussed, for instance, in Refs. \citen{pHal11,bHal16}. In these references, freedom is not defined as equivalent to MI, as we did.
There, the concept of freedom only requires $p(a,b)\neq0$ so that $a$ and $b$ are both eligible.
Here we define counterfactual freedom by $p(a_i,b_k|\lambda)\neq0,\,i,k\in\{1,2\}$ for all values of $\lambda$.
Thus, CF does not require (\ref{eq:mi}) and we may have $p(\lambda|a,b)\neq p(\lambda)$.
\section{The Physical Irrelevance of CFD}\label{sec:TPICFD}
We present two theorems proving that CFD is physically irrelevant, as it does not lead to a falsifiable result.
Furthermore, although we shall not prove it here, it can be shown that when the BI is derived through CFD, it cannot be meaningfully compared with the quantum mechanical predictions \cite{pLam20a}.
\begin{theorem}\label{t:t1}
The realism hypothesis(R) does not imply counterfactual definiteness
\begin{eqnarray}
\neg(R &\longrightarrow CFD)\label{leq:t1}
\end{eqnarray}
\end{theorem}
\begin{proof}
We base the proof of the theorem on Michel Feldmann's local realistic model.\footnote{Feldmann's model \cite{pFel95} is a plane model and was later generalized to three dimensions\cite{pDeg05}. In these models, the hidden variable is the angle determining the spin direction.} This model violates the Bell inequality by transgressing MI.

Adapting Feldmann's notation to the one we used in sec. \ref{ssec:BSD} and restricting spin measurements to planes orthogonal to the particles motion with $\lambda\in[0,2\pi]$
\begin{eqnarray}
A(a_i,\lambda)     & = & + sgn(\cos(\lambda-a_i))\\
B(b_i,\lambda)     & = & - sgn(\cos(\lambda-b_i))\\
p(\lambda|a_i,b_k)=p(\lambda|u)    & = & \frac{1}{4}|\cos(\lambda-u)|,\quad u\in\{a_i,b_k\}\label{eq:felrho}
\end{eqnarray}
The model is local because the $\lambda$ variables correspond to events in the common causal or past light cones of A and B.
Besides, neither $A(a_i,\lambda)$ nor $B(b_k,\lambda)$ depend on the other's setting, therefore they represent realistic functions completely determined by local factors and the common causes represented by $\lambda$.

In (\ref{eq:felrho}) either $u=a_i$ or $u=b_k$, thus in Feldmann's model $p$ depends only on one setting.
For instance, we can put $u=a_1$ when Alice uses setting $a_1$ no matter what setting Bob uses, similarly when Alice chooses setting $a_2$.
So, there are no ambiguities because his consistency equations are fulfilled
\begin{eqnarray}
\langle A_iB_k\rangle=E(a_i,b_k) &=& \int_0^{2\pi}A(\lambda,a_i)B(\lambda,b_k)p(\lambda|a_i)\,d\lambda\nonumber\\
                                 &=& \int_0^{2\pi}A(\lambda,a_i)B(\lambda,b_k)p(\lambda|b_k)\,d\lambda\label{eq:felcor}
\end{eqnarray}
Although we can evaluate the correlation $E(a_i,b_k)$ using either $a_i$ or $b_k$ in $p$, both give the same result, so  no ambiguity is produced.
The same is true for,
{\small
\begin{equation}
\langle A_i\rangle=E(a_i)=\int_0^{2\pi}A(\lambda,a_i)p(\lambda|a_i)d\lambda = \int_0^{2\pi}A(\lambda,a_i)p(\lambda|b_k)d\lambda=0
\end{equation}
\begin{equation}
\langle B_k\rangle=E(b_k)=\int_0^{2\pi}B(\lambda,b_k)p(\lambda|b_k)d\lambda = \int_0^{2\pi}B(\lambda,b_k)p(\lambda|a_i)d\lambda=0
\end{equation}
}
It is easy to compute (\ref{eq:felcor}) and find
\begin{eqnarray}
E(a_i,b_k)=- cos(a_i-b_k)\label{eq:qmcor}
\end{eqnarray}
Thus, Feldmann's model is local realistic and reproduces the quantum mechanical correlations of the singlet state thus violating the CHSH inequality; really, for certain appropriate settings it is known that (\ref{eq:qmcor}) implies
\begin{eqnarray}
|S|=2\sqrt{2}\label{ap2:slim}
\end{eqnarray}
Although Feldman's model infringes MI, it respects CF.
In fact, there is nothing in the model imposing  $p(a_i,b_k|\lambda)=0$ for some $\lambda$ with $i,k\in\{1,2\}$, therefore we can safely assume that all four settings possibilities are counterfactually accessible.
\footnote{From Bayes's theorem $p(\lambda|a_i,b_k)p(a_i,b_k)=p(a_i,b_k|\lambda)p(\lambda)$ and, according to (\ref{eq:felrho}), the first factor in the left hand vanishes at four values $\lambda-a_i=\pi/2(Mod\,2\pi)$  and $\lambda-b_k=\pi/2(Mod\,2\pi)$. However, we can have $p(a_i,b_k|\lambda)\neq0$ by requiring $p(\lambda)=0$ at those points.
}

Finally, we have that Feldmann's model complies with realism(R), locality(L), and counterfactual freedom(CF).
Since it violates the Bell inequality then, according to (\ref{leq:CBT}), we cannot have CFD notwithstanding the presence of realism.
\begin{eqnarray}
\neg(R \longrightarrow CFD)\nonumber
\end{eqnarray}
\end{proof}
%
\begin{theorem}\label{t:t2}
When the BI is derived by assuming CFD, L, and F, experimental violations of the inequality cannot discard local realism.
\end{theorem}
\begin{proof}
According to (\ref{leq:fandcf}), CF is valid, then by (\ref{leq:CBT})
\begin{eqnarray}
CFD\wedge L\wedge F\rightarrow BI
\end{eqnarray}
When experiments violate the BI, if we retain L and F, we must reject CFD then, as a consequence of theorem \ref{t:t1}, we cannot infer that realism is false; therefore, we cannot discard local realism.
\end{proof}
\section{The Physical Interpretation of Bell's Derivation}\label{sec:exprot}
As we mentioned in sec. \ref{ssec:BSD}, equation (\ref{eq:cs}) is a common source of confusion because each term is the product of results of a joint measurement performed on each particle of an entangled pair.
Therefore, (\ref{eq:cs}) contains the results of four different experiments. However, the four different experiments are supposed to contain the same hidden variable $\lambda$.

How can that be? The experimenter has no control over the hidden variables and we do not even know whether they actually exist.
All we know is that each singlet pair is supposed to be generated with a given hidden variable value.
Some have proposed to make four successive measurements on the same entangled particle pair\cite{pDPCB72}.
Since it is not possible to perform more than one measurement on a particle without disturbing it,\footnote{In the case of photons, the particles are annihilated after they are absorbed.} the most popular solution is to interpret (\ref{eq:cs}) according to the CFD recipe: only one experiment is actually performed, the other three represent the results of experiments that could have been performed but were not.

However, when we follow the derivation steps in the correct order starting by (\ref{eq:s1}), such an interpretational conundrum does not arise.
\footnote{A common error inducing practice starts the derivation by (\ref{eq:cs}), see Ref. \citen{pLam17a}.}

Really, to falsify the theoretical result, all we need is to measure the four correlation terms in (\ref{eq:s1}).
The obtention of (\ref{eq:s1}) -- with actual values -- only requires the repetition of individual experiments measuring  the ``clicks''  $A'(a_i)$ and $B'(b_k)$ for each singlet state.
Real experiments usually record the products $A_{ik}=A'(a_i)B'(b_k)$, not the individual ``clicks'' but this is not important.

Table \ref{tabla:t1} shows a summary of the actual data that would be obtained in an idealized experiment with $100\%$ percent detection efficiency.
\begin{table}
\caption{Experimental Results\label{tabla:t1}}
\begin{center}
\begin{tabular}{c|c|c|c|c|c}
Event\# & $A_1B_1$  & $A_1B_2$ & $A_2B_1$ & $A_2B_2$ & $\lambda$\\
\hline
$1$     &   $+1$     & ***     &   ***     &    ***    & unknown\\
$2$     &   ***     & $-1$     &   ***     &    ***    & unknown\\
$3$     &   ***     & ***      &   ***     &    $-1$   & unknown\\
$\vdots$& $\vdots$  &$\vdots$   & $\vdots$    &$\vdots$     &$\vdots$
\end{tabular}
\end{center}
\end{table}
Notice that, so far, we have not move from (\ref{eq:s1}) in the derivation and we have obtained with actual measurements the value $|S|$. The experimental result found for $|S|$ falsifies what we have assumed in steps (\ref{eq:s2}) through (\ref{eq:s6}) of the derivation.

We proved in sect. \ref{sec:TPICFD} that, although we can predict the results of impossible experiments, if we use unrealizable predictions, the proof becomes physically meaningless. We must, instead, explain the presumed emergence of (\ref{eq:cs}) through a physically meaningful mechanism:
\begin{quote}
\textbf{Physical Interpretation of (\ref{eq:cs})}: After the experiment has been run for a sufficiently long time, the values of $\lambda$ are supposed to be randomly and uniformly repeated for the different settings used in the experiment. This  constitutes a statistical regularity assumption that Willy De Baere\cite{pDBa84a} termed the \emph{reproducibility hypothesis}.\footnote{De Baere's hypothesis is not an independent assumption. It can be shown to be a consequence of the hidden variables set finiteness and of MI\cite{pLam20a}.}
\end{quote}
The former interpretation is what physically justifies the mathematical step from (\ref{eq:s1}) to (\ref{eq:s2}) and the emergence of (\ref{eq:cs}).
It allows the extraction of $p$ as a common factor in step (\ref{eq:s2}) and constitutes the concrete mechanism that would validate and explain the emergence of the infamous (\ref{eq:cs}) with data that would be obtained in actual experiments, falsifying local realism.

On the other hand, when we base our derivation on CFD, theorem \ref{t:t2} proves the violation of the inequality cannot falsify local realism.
The last result should not be surprising.
When we base our derivation on impossible experiments, the comparison of our theoretical results with the experiments is meaningless.

Curiously, the point seems to be very hard to grasp.
That is why, in sect. \ref{sec:TPICFD}, we chose a mathematical proof that we hope would end decades of faulty arguments and interpretational discussions.
The problem is related to the subtle difference between predictability and testability.
\section{Conclusions}
John Bell's notable breakthrough was to replace EPR's thought experiment by another one which does not involve irreproducible situations so that it could be experimentally tested. Surprisingly, this turned out to be a very subtle point.
As Bell expressed concerning determinism\cite{pBel81}, it is remarkably difficult to get this point across, that incompatible experiments are not a presupposition of the analysis.

Although counterfactual reasoning is a valid principle for philosophical discussions, theorem \ref{t:t2} proves that when it is used to prove the BI, the physically relevant properties become logically disconnected from the experimental results.
This means that, although a theory can make counterfactual predictions, when such metaphysical speculations are confronted with actual facts, logical inconsistencies may turn out to be unavoidable.

Unlike the EPR reasoning, the Bell inequality is a testable theoretical prediction concerning what happens in actually performed experiments, not a philosophical contemplation about results of impossible experiments mistakenly assumed falsifiable.
\section*{Acknowledgements}
The author wants to express his gratitude to Dr. Michael Hall for useful discussions on subjects related to this manuscript.
%
%
\bibliography{BTCFD-arxiv-8}
\end{document}